\documentclass[letterpaper, 10 pt, conference]{ieeeconf}  

\IEEEoverridecommandlockouts                              
\usepackage{mathrsfs}
\usepackage{amsmath} 
\usepackage{amssymb}  
\usepackage{cite}
\usepackage{amsthm}
\usepackage{enumerate}
\usepackage{subcaption}
\usepackage{graphicx}
\usepackage{mathtools, cuted}
\usepackage{stfloats}
\usepackage{comment}
\usepackage{xcolor}
\usepackage[linesnumbered,ruled,vlined]{algorithm2e}

\newcommand{\R}[1]{\mathbb{R}^{#1}}
\newcommand{\linf}[1]{\left\|#1 \right\|_{\infty}}
\newcommand{\Ex}[1]{\mathcal{E}\left( #1\right)}
\newcommand{\Pt}[1]{\mathcal{P}_{#1}}

\newcommand{\lemref}[1]{(Lem.\ref{#1})}
\newcommand{\thmref}[1]{(Thm.\ref{#1})}
\newcommand{\corref}[1]{(Coro.\ref{#1})}
\newcommand{\propref}[1]{(Prop.\ref{#1})}
\newcommand{\assref}[1]{(Ass.\ref{#1})}
\newcommand{\secref}[1]{(Sec.\ref{#1})}
\newcommand{\conref}[1]{(Const.\ref{#1})}
\newcommand{\algoref}[1]{(Algo.\ref{#1})}
\newcommand{\figref}[1]{(Fig.\ref{#1})}
\newcommand{\defref}[1]{(Def.\ref{#1})}

\theoremstyle{plain}

\theoremstyle{definition}
\newtheorem{coro}{Corollary}[section]
\newtheorem{thm}{Theorem}[section]
\newtheorem{lem}{Lemma}[section] 
\newtheorem{prop}{Proposition}

\newtheorem{defn}[thm]{Definition} 
\newtheorem{rem}{Remark}
\newtheorem{asp}{Assumption}
\newtheorem{con}{Constraint}

\newcommand{\up}[2]{#1 \leftarrow #2}
\renewcommand{\c}[1]{\mathcal{#1}}
\excludecomment{lv}
\includecomment{sv}
\newenvironment{col}{\color{blue}}{}

\title{\LARGE \bf
Scalable Robust Adaptive Control from the System Level Perspective\\
 }

\author{Dimitar Ho and John C. Doyle
\thanks{Dimitar Ho and John C. Doyle are with the Department of Computing and Mathematical Sciences, California Institute of Technology, Pasadena, CA. 
        {\tt\small dho@caltech.edu, doyle@caltech.edu}}%
}

\begin{document}

\maketitle
\thispagestyle{empty}
\pagestyle{empty}

\begin{abstract}

\noindent We will present a new general framework for robust and adaptive control that allows for distributed and scalable learning and control of large systems of interconnected linear subsystems. The control method is demonstrated for a linear time-invariant system with bounded parameter uncertainties, disturbances and noise. The presented scheme continuously collects measurements to reduce the uncertainty about the system parameters and adapts dynamic robust controllers online in a stable and performance-improving way. A key enabler for our approach is choosing a time-varying dynamic controller implementation, inspired by recent work on \textit{System Level Synthesis} \cite{slsacc}. We leverage a new robustness result for this implementation to propose a general robust adaptive control algorithm. In particular, the algorithm allows us to impose communication and delay constraints on the controller implementation and is formulated as a sequence of robust optimization problems that can be solved in a distributed manner. The proposed control methodology performs particularly well when the interconnection between systems is sparse and the dynamics of local regions of subsystems depend only on a small number of parameters. As we will show on a five-dimensional exemplary chain-system, the algorithm can utilize system structure to efficiently learn and control the entire system while respecting communication and implementation constraints. Moreover, although current theoretical results require the assumption of small initial uncertainties to guarantee robustness, we will present simulations that show good closed-loop performance even in the case of large uncertainties, which suggests that this assumption is not critical for the presented technique and future work will focus on providing less conservative guarantees.

\end{abstract}


\section{INTRODUCTION}
\begin{lv}
\begin{col}
The question of how to control systems, for which one does not have an accurate model, has been of central interest to the control community for many decades with some early work dating back as early as 1950 by Kalman and Bellman. Roughly speaking, the control methodologies developed to tackle this problem can be classified in two, partially overlapping groups: robust control and adaptive control. In robust control, one aims to design a controller that can guarantee a certain performance for all possible systems within an uncertainty class. Whereas adaptive control parametrizes the system uncertainties and studies how to estimate the corresponding parameters and adapt the feedback controller online. Finally, merging the former two frameworks is the focus of robust adaptive control. For a more detailed overview of adaptive control, we refer to the survey \cite{tao2014multivariable} and \cite{ioannou1996robust}, as we can not review the rich history of this subject here.
\end{col}
\end{lv}

\noindent With the recent explosion of available computational resources and progress in the field of learning and estimation theory, there has been a resurging interest in robust adaptive control in the control and also machine learning community. In contrast to traditional work in \cite{tao2014multivariable} and \cite{ioannou1996robust}, recent work has focused on analysis and development of adaptive control algorithms that merge learning and statistical theory techniques \cite{dean2017sample}, \cite{dean2018regret}, \cite{abbasi2018regret}. Although adaptive control algorithms are very useful for many systems of large-scale like communication networks, traffic networks or the power grid, there has not been a general theory of how to address the challenges in that setting. One of the major difficulties with deploying scalable adaptive algorithms in systems of that scale is, that the controller has to respect real-world implementation and communication constraints. Even in the non-adaptive case, incorporating these constraints into the control design is a challenging problem. Nevertheless, recent progress has been made by taking a new \textit{System Level} approach \cite{slsacc},\cite{wang2014localized}, \cite{virtloc}, that allows to incorporate such constraints into optimal control problems in a tractable way. Aside from that, recent work  \cite{matni2017scalable}, \cite{dean2017sample} has shown that the ideas in \cite{virtloc} can be used to provide robustness results that help to combine learning and control techniques with stability guarantees.\\

\noindent In this work, we will leverage the system level approach to formulate a new general framework for robust adaptive control in large-scale systems. In particular, we will study the problem for linear systems with bounded uncertainty and disturbances. An appeal of this problem formulation is that in contrast to probabilistic guarantees as formulated in the results of \cite{dean2018regret}, \cite{abbasi2018regret}, we are able to provide worst-case safety guarantees that apply even in the presence of adversarial disturbances and small model non-linearities. 
 Overall, the contribution of this paper is two-fold: We will derive robustness criteria similar to \cite{virtloc} for time-varying systems and controllers that provide a new general way to design stable adaptation in controllers. Secondly, we utilize these results to develop a robust and adaptive control scheme that can respect imposed communication and implementation constraints on the controller and allows for a distributed scalable implementation in large scale systems.  Although our current stability proof is formulated for small initial uncertainties, in simulation we will show that the resulting control algorithm performs well even when the initial parameter uncertainties are large and the open loop system is unstable.

\section{A Motivational Example: A 5-link chain system}\label{sec:introex}
\noindent We will begin by introducing the example which we use for our simulation results, to motivate the problem statement and the techniques presented in this work. \\

\noindent Consider the problem of controlling the following $5$-link chain-system  \eqref{eq:5chainsys} with the state $x_t \in \R{5}$, input $u_t \in \R{2}$, disturbance $w_t \in \R{5}$ and full state measurement $y_k = x_k$:
\begin{align}\label{eq:5chainsys}
x_{t} &= A x_{t-1} + B u_{t-1} + w_{t-1} \\
y_{t} &= x_t\nonumber
\end{align}
Furthermore, assume that we do not have exact knowledge of $A$ and $B$, but rather we do know that the system matrices $A$ and $B$ are structured as
\begin{align}\label{eq:ABstruct}
A = \begin{bmatrix} \alpha_2 &\alpha_3&0&0&0\\ 
  \alpha_1& \alpha_2 &\alpha_3&0&0\\
0&  \alpha_1& \alpha_2 &\alpha_3&0\\ 
0&0&  \alpha_1& \alpha_2 &\alpha_3\\
0&0&0&  \alpha_1& \alpha_2 \end{bmatrix} & B = \begin{bmatrix} \alpha_4 & 0\\ 
 0& 0\\
0&0\\ 
0&0\\
0&\alpha_5 \end{bmatrix} 
\end{align}
and that the parameters $\alpha$ lie within the bounds $ 0 \leq \alpha_2 \leq 1$, $ 0.1 \leq \alpha_1,\alpha_3 \leq 0.5$, $0.2 \leq \alpha_4 \leq 1$ and $-1 \leq \alpha_5 \leq -0.2$. In addition, assume that we know that the disturbance is bounded as $\linf{w} \leq 0.5$.\\ 

\noindent We are interested in the problem of stable learning and control of this system under communication and computation constraints on the controller implementation. In particular, we will assume that \eqref{eq:5chainsys} models a system of five interconnected, but otherwise separately acting scalar subsystems with state $x^i_t$ and input $u^i_t$ (where $u^i_t = 0$ for $i = 2,3,4$), that can communicate with eachother with a delay of $|i-j|$ time steps and each have limited computational power. 
 
  Although this example is of small size, this problem setup captures the main difficulties that come with solving this type of robust adaptive control problem for large-scale systems, which will be the focus of the remainder of this paper.




\section{Problem Statement}\label{sec:prob}
\noindent Consider a linear system of $N$ subsystems with states, inputs and disturbances $x^j$, $u^j$, $w^j$ that are interconnected w.r.t. the directed graph $\c{G} = (\c{V},\c{E})$, i.e. $\c{V}= \left\{1,\dots,N \right\}$ and $(j,i) \in \c{E}$ implies that $x^{i}_t$ influences $x^{j}_{t+1}$. Furthermore, define $\c{N}(j)$ to be the set of subsystems that affect the subsystem $j$ in the next time step, i.e. $\c{N}(j) = \left\{i\left|(j,i) \in \c{E} \right. \right\}$. The dynamics of the entire system can be written in the form
\begin{align}\label{eq:sysdist}
x^{j}_{t+1} & = \sum\limits_{i \in \c{N}(j)}A^{\up{j}{i}} x^i_t + B^j u^j_t + w^j_t.
\end{align}
and we will refer to $x_t = \left[x^{1}_t,x^{2}_t,\dots,x^{N}_t\right]^T$ and $u_t = \left[u^{1}_t,u^{2}_t,\dots,u^{N}_t\right]^T$ as the global state and input of the system and accordingly, we will refer to $A$ and $B$ as the global system matrices, which are the corresponding compositions of the matrices $A^{\up{j}{i}}$ and $B^j$.
\begin{rem}
We allow for loops in the graph $\c{G}$, which implies $j \in \c{N}(j)$.
\end{rem}
\noindent Similar to our introductory example in \secref{sec:introex}, we will assume that the matrices $A^{\up{j}{i}}$ and $B^j$ are structured and have a low-dimensional representation of the form \eqref{eq:AvuBu} w.r.t. some uncertain parameters $\alpha \in \R{p}$ and known constant matrices $\c{A}^{\up{v}{u}}_s$, $\c{B}^{u}_s$.
\begin{align}\label{eq:AvuBu}
 &A^{\up{v}{u}} = \sum\limits^{p}_{s=1} \alpha_s \c{A}^{\up{v}{u}}_s & B^u = \sum\limits^{p}_{s=1} \alpha_s\c{B}^{u}_s
 \end{align}
 Furthermore, assume we are given the following information about the parameter $\alpha$ and the disturbance $w^j$ in each subsystem:
 \begin{align}\label{eq:previnfo}
& \alpha \in \c{P}_0  & \left\|w^j_t\right\| \leq \eta \quad \forall t\geq 0
 \end{align}
Moreover we can setup the problem with any norm $\left\|.\right\|$, but for technical reasons, we will make the following assumption:
 \begin{asp}\label{asp:norm}
The unit ball $\left\{x \left| \left\|x\right\| \leq 1 \right. \right\}$ of the norm is a polytope. 
 \end{asp}
 \begin{rem}
 Common examples that satisfy \assref{asp:norm} are $\left\|.\right\|_{1}$ and $\left\|.\right\|_{\infty}$.
 \end{rem}

\subsection {Main Goal} 
\noindent Our objective will be to design causal controllers $u^{j}_k(x_k,x_{k-1},\dots, x_0)$ that stabilize the global system despite the model uncertainties and allow for a scalable controller implementation when the total number of subsystems is very large. To make the second requirement more precise, we will break it down into the following three constraints:
 
 
\begin{con}[Communication]\label{con:delayintro}
Every subsystem $j$ can communicate with another subsystem $i$ with a delay of $d^{\up{j}{i}}$ time-steps.
\end{con}
 
\begin{con}[Localized Communication]\label{con:locintro}
Every subsystem $i$ only sends information to a local region of subsystems $\c{S}(i)$.
\end{con}
\noindent Corresponding to \conref{con:locintro}, let's define $\c{R}(i)$ to be the set of subsystems from which subsystem $i$ receives information:
\begin{defn}\label{def:Ri}
 $\c{R}(i) := \left\{j \left| i \in \c{S}(j)\right. \right\}$
\end{defn}
\begin{con}[Limited Computation]\label{con:loccompintro}
Every subsystem $j$ has limited computational resources.
\end{con}

\section{Outline of the Approach}
\noindent We will briefly motivate our chosen control architecture and provide an overview of the results.
\subsection{Ansatz: Time-Varying Controllers in SLS Implementation}

\noindent The recent SLS approach \cite{slsacc} develops an useful insight for linear time-invariant systems: An equivalence relationship between the closed loop maps from disturbance to state/control action and the corresponding realizing controller. In particular, assume we are controlling system \eqref{eq:systi} with a time-invariant controller $u$
\begin{align}\label{eq:systi}
x_{k+1} = Ax_k + Bu(x_k,x_{k-1},\dots) + w_k
\end{align}
that guarantees $x_{T+1} = 0$ for any initial condition and $w_k=0$. Then the corresponding closed loop map $w_k \rightarrow x_k$ and $w_k \rightarrow u_k$ can be written in the form
\begin{align}\label{eq:clmap}
&x_t  = \sum\limits^{T}_{k=1} R(k)w_{t-k} &u_t  = \sum\limits^{T}_{k=1} M(k)w_{t-k}
\end{align} 
for some matrices $R(k)$, $R(1) = I$ and $M(k)$. Then, as proven in \cite{slsacc}, the realizing controller $u(x_k,x_{k-1},\dots)$ can equivalently be implemented through the equations
\begin{align*}
\hat{\delta}_t &= x_t-\sum^{T-1}_{k=1} R(k+1)\hat{\delta}_{t-k}\\
u_t &= \sum^{T-1}_{k=0} M(k+1)\hat{\delta}_{t-k}
\end{align*}
where $\hat{\delta}_t$ is the internal state of the controller. In addition, as shown in \cite{virtloc}, this controller implementation provides stable closed loop systems even when $R$ and $M$ satisfy the relation \eqref{eq:clmap} only approximately. With this motivation we will structure our controller to be in \text{SLS implementation} and take the form:
\begin{align}\label{eq:usls}
u_t &= \sum^{T-1}_{k=0} M_t(k+1)\hat{\delta}_{t-k}\\
\label{eq:dsls} \hat{\delta}_t &= y_t-\sum^{T-1}_{k=1} R_{t}(k+1)\hat{\delta}_{t-k}.
\end{align}
where $M_t(i) \in \R{p\times n}$, $R_t(i)\in \R{n\times n}$, $\forall 1 \leq i\leq T$ and $R(1) = I_n$. \\

\noindent Our robust adaptive control scheme will propose algorithms that continuously use state observations to update the $R_t$ and $M_t$ matrices in a stable and performance improving manner. More specifically, this happens in two steps: First, we utilize the system equations \eqref{eq:sysdist} and the disturbance assumption \eqref{eq:previnfo} to continuously compute polytopes of possible parameters $\alpha$ from observations. This parameter information is then used to stably adapt the matrices $R_t$ and $M_t$ by utilizing a new result on robustness of time-varying controllers in SLS implementation form.\\

\noindent In the following sections, we will elaborate upon these ideas. In the next two sections, we will show how to compute parameter polytopes from observations that continuously reduce the uncertainty in $\alpha$ and we will derive a robustness result for time-varying controllers of the form \eqref{eq:usls}, \eqref{eq:dsls}. Afterwards we will discuss how these ideas can be combined to formulate a robust control scheme that allows for scalable implementation in large-scale systems and follow up with simulation results of the exemplary system introduced in \secref{sec:introex}.

\section{Reducing Uncertainty through Polytopes of Consistent Parameters}\label{sec:poly}
\noindent Recall from \eqref{eq:AvuBu}, that $A^{\up{v}{u}}$ and $B^u$ are structured as
\begin{align}
 &A^{\up{v}{u}} = \sum\limits^{p}_{s=1} \alpha_s \c{A}^{\up{v}{u}}_s & B^u = \sum\limits^{p}_{s=1} \alpha_s\c{B}^{u}_s
 \end{align}
 with $\alpha \in \Pt{0}$. Given a pair of observation $x_{k}, x_{k-1}$ and control action $u_{k-1}$, the disturbance bound $\left\| w^j\right\| \leq \eta$ informs us about $\alpha$, since the true $\alpha$ has to be consistent with the dynamics \eqref{eq:sysdist} and therefore has to satisfy the following inequality:
\begin{align}\label{eq:s1}
\left\| x^j_{k}-\sum^{p}_{s=1}\alpha_s \left( \sum\limits_{i\in \c{N}(j)} \c{A}^{\up{j}{i}}_s x^{i}_{t-1} + \c{B}^{j}_{s} u^{j}_{t-1} \right) \right\|\leq \eta
\end{align}
or equivalently
\begin{align}\label{eq:s2}
\left\| x^j_{k}-\sum^{p}_{s=1}\alpha_s \hat{y}^{j}_{s,k-1} \right\|\leq \eta\\
\label{eq:ys}\hat{y}^{j}_{s,k-1}  = \sum\limits_{i\in \c{N}(j)} \c{A}^{\up{j}{i}}_s x^{i}_{t-1} + \c{B}^{j}_{s} u^{j}_{t-1}
\end{align}
Since we assumed in \assref{asp:norm} that the norm ball of $\left\|.\right\|$ is a polytope, it is easy to see that condition \eqref{eq:s2} poses a polyhedral constraint on the system parameters $\alpha$. We will define these inferred constraints from observations made in subsystem $j$ at time $t$ as $\c{C}^j_t$:
\begin{align}\label{eq:Ctcons}
\c{C}^j_t = \left\{\alpha\left|\left\| x^j_{k}-\sum^{p}_{s=1}\alpha_s \hat{y}^j_{s,k-1} \right\|\leq \eta \right. \right\}
\end{align}
By intersecting all constraints of the form \eqref{eq:Ctcons}, we can define $\Pt{t}$ as the polytope of parameters consistent with the observations until time $t$:
\begin{align}\label{def:Ptj}
\notag &\Pt{t} \\
=& \left\{\left.\alpha \in \Pt{0}\right| \forall j, \forall k \leq t: \left\| x^j_{k}-\sum^{p}_{s=1}\alpha_s \hat{y}^j_{s,k-1} \right\|\leq \eta\right\}\\
=&\Pt{0} \cap \bigcap\limits^{N}_{j=1} \left( \c{C}^j_1 \cap \c{C}^j_2 \cap \dots \cap \c{C}^j_t\right)
\end{align}
Correspondingly define $\c{M}^{\up{j}{i}}_A\left( \Pt{t}\right)$  and $\c{M}^{j}_B\left( \Pt{t}\right)$ to be the set of consistent system matrices $A^{\up{j}{i}}$ and $B^j$  at time $t$:
\begin{align}
\label{eq:MofPA} \c{M}^{\up{j}{i}}_A\left( \Pt{t}\right) &= \left\{\left.\sum^{p}_{s = 1} \alpha_s \c{A}^{\up{j}{i}}_s \right| \alpha \in \Pt{t} \right\} \\
\label{eq:MofPB}\c{M}^{j}_B\left( \Pt{t}\right) &= \left\{\left.\sum^{p}_{s = 1} \alpha_s \c{B}^{j}_s \right| \alpha \in \Pt{t} \right\}
\end{align}

\noindent Furthermore, allowing every subsystem to share their observed constraints while respecting \conref{con:delayintro} and  \conref{con:locintro}, we can define $\Pt{t}^j$ as the polytope of consistent parameters for subsystem $j$ at time $t$ as:
\begin{align}\label{eq:Pi_t}
&\c{P}^j_t = \c{P}^j_{t-1} \bigcap_{i \in \c{R}(j)} \c{\hat{C}}^{\up{j}{i}}_t &\c{\hat{C}}^{\up{j}{i}}_t = \c{C}^i_{t-d_{\up{j}{i}}}
\end{align} 
where $\c{\hat{C}}^{\up{j}{i}}_t$ denotes the constraints that $j$ has obtained from system $i$ at time $t$ and $\c{R}(j)$ is defined in \defref{def:Ri}.

\section{Robustness for Time-Varying SLS Implementations}\label{sec:slst}
\noindent In this section we will develop a robustness result which directly informs us how to perform stable adaptation of $R_t$ and $M_t$.\\ For the sake of completeness, let us consider the general time-varying linear system 
\begin{align}\label{eq:syst}
x_{t} &= A_{t-1} x_{t-1} + B_{t-1} u_{t-1} + w_{t-1}\\
y_t &= x_t + v_t
\end{align}
with $u$ being a controller of the form \eqref{eq:usls} and \eqref{eq:dsls}. The internal state $\hat{\delta}$ can be understood as an estimate of the \textit{effective disturbance} in the system at time $t$ and by rewriting \eqref{eq:dsls} as 
\begin{align}\label{eq:xdef}
x_t = \sum^{T-1}_{k=0} R_{t}(k+1)\hat{\delta}_{t-k} - v_t
\end{align}
we see that $x_t$ is bounded if and only if this effective disturbance $\hat{\delta}$ is bounded. Using \eqref{eq:syst}, we can conclude that the dynamics of $\hat{\delta}$ satisfy the equations:
\begin{align}
\hat{\delta}_{t} &= x_t+v_t-\sum^{T-1}_{k=1} R_{t}(k+1)\hat{\delta}_{t-k}\\
&=A_{t-1}\left(\sum^{T}_{k=1} R_{t-1}(k)\hat{\delta}_{t-k} - v_{t-1}\right) \dots \nonumber \\
&\quad + B_{t-1}\left(\sum^{T}_{k=1} M_{t-1}(k)\hat{\delta}_{t-k}\right)\dots \nonumber\\
&\quad -\sum^{T-1}_{k=1} R_{t}(k+1)\hat{\delta}_{t-k} + v_t + w_{t-1}. 
\end{align}
Introduce the functions $\Delta_1$, $\Delta_2$, $\dots$, $\Delta_T$ as
\begin{align}
\notag\Delta_k\left(A,B,R,M\right) &= R(k+1)-AR(k)-BM(k)\\
\label{eq:Ddef}\Delta_T\left(A,B,R,M\right) &= -AR(T)-BM(T)
\end{align}
then we can write the dynamics for $\hat{\delta}$ in the form:
\begin{align}\label{eq:ddyn}
\notag\hat{\delta}_t&= -\sum^{T}_{k=1} \Delta_k\left(A_{t-1},B_{t-1},R_{t-1},M_{t-1}\right)\hat{\delta}_{t-k} \\
\notag& \quad \dots + \sum^{T-1}_{k=1} \left(R_{t-1} -R_{t}\right)(k+1)\hat{\delta}_{t-k} \\
&\quad \quad \dots + \left(v_t -A_{t-1}v_{t-1}\right) + w_{t-1}.
\end{align}
It is now clear that making the closed loop system stable is equivalent to keeping the dynamics of the effective disturbance in \eqref{eq:ddyn} bounded.\\

\noindent To this end, the following stability result for scalar systems will be useful: 

\begin{lem}\label{lem:scalar}
Consider the positive scalar sequence $z_t$ that satisfies
\begin{align}\label{eq:zcond}
&z_t \leq \lambda \max\limits_{1 \leq k\leq t} z_{t-k} + \eta &\text{ for } 1\leq t\leq T\\   
 &z_t \leq \lambda \max\limits_{1 \leq k\leq T} z_{t-k} + \eta &\text{ for } t > T
\end{align}
with  $0< \lambda $, then the following bound holds true $\forall t \geq 0$:
\begin{align}
\label{eq:zboundstab}& z_t \leq \left(\sqrt[T]{\lambda}\right)^t z_{0}+\frac{1-\lambda^t}{1-\lambda}\eta & \text{for }\lambda <1\\
\label{eq:zboundtriv} & z_t \leq \lambda^t z_{0}+\frac{1-\lambda^t}{1-\lambda}\eta & \text{for }\lambda \geq 1
\end{align}
\end{lem}
\begin{proof}
We will only derive \eqref{eq:zboundstab}, since \eqref{eq:zboundtriv} follows trivially. 
Pick any trajectory that satisfies \eqref{eq:zcond} and fix $t'$. By using the relationship \eqref{eq:zcond}, we will construct a subsequence $z_{t_j}$, with $j=0,\dots,N$, $t_0 = 0$, $t_N = t'$ such that 
\begin{align}\label{eq:t_rec}
z_{t_{n-1}} &= \max\limits_{\min\left\{t_{n}-T, 0\right\} \leq \tau \leq t_{n}-1} z_{\tau}.
\end{align}
Starting from $t_N = t'$ and proceeding backwards according \eqref{eq:t_rec}, it is clear that this  construction is always possible for some $N$ and by examining the recursive relation \eqref{eq:t_rec}, we can bound $N$ by 
\begin{align}\label{eq:Nbnd}
t'/T \leq N \leq t'
\end{align}  
The constructed subsequence satisfies for all $1\leq j\leq N$ the relation
\begin{align}
z_{t_j} \leq \lambda  z_{t_{j-1}}+\eta.
\end{align} 
Per induction we can derive the following bound for $z_{t'} = z_{t_N}$:
\begin{align}\label{eq:gamN}
z_{t'} = z_{t_N} \leq \lambda^Nz_{t_0} + \frac{1-\lambda^N}{1-\lambda}\eta =\lambda^Nz_{0} + \frac{1-\lambda^N}{1-\lambda}\eta
\end{align}
In terms of $N$, the first summand in \eqref{eq:gamN} is monotonically decreasing, while the second one is monotonically increasing. Therefore, using \eqref{eq:Nbnd}, we can further bound $z_{t'}$ by
\begin{align}\label{eq:gambnd}
z_{t'}  \leq \left(\sqrt[T]{\lambda}\right)^{t'}z_{0} + \frac{1-\lambda^{t'}}{1-\lambda}\eta. 
\end{align}
Since the time-step $t'$ was chosen arbitrarily, \eqref{eq:gambnd} establishes the desired result.
\end{proof}

\noindent For any norm $\left\|x\right\|$, let $\left\|A\right\|$ be the corresponding induced norm for matrices $A$. 
We will use the abbreviation $\hat{w}_t := \left(v_t -A_{t-1}v_{t-1}\right) + w_{t-1}$ and $\left\|\hat{w}_t \right\| \leq \hat{\eta}$ and obtain from \eqref{eq:ddyn} the bound
\begin{align}
   \left\|\hat{\delta}_t\right\| &\leq \sum^{T}_{k=1} \left\|\Delta_k\left(A,B,R,M\right)_{t-1}\right\| \left\|\hat{\delta}_{t-k}\right\| + \dots \\
&\quad\left\|\sum^{T-1}_{k=1} \left(R_{t-1} -R_{t}\right)(k+1)\hat{\delta}_{t-k}\right\| + \hat{\eta}.
\end{align}
and furthermore
\begin{align}
   \left\|\hat{\delta}_t\right\| &\leq \left(\sum^{T-1}_{k=1} \left\|\Delta_k\left(A,B,R,M\right)_{t-1}\right\|\right) \max\limits_{1\leq k\leq T}\left\|\hat{\delta}_{t-k}\right\| + \dots \\
&\quad\left\|\sum^{T-1}_{k=1} \left(R_{t-1} -R_{t}\right)(k+1)\hat{\delta}_{t-k}\right\| + \hat{\eta}.
\end{align}

The following robustness result merely follows as a corollary of the previous Lemma \lemref{lem:scalar}:
\begin{thm}\label{thm:slst}
Assume that $A_t$, $B_t$, $M_t$, $R_t$ are bounded matrix sequences. Consider the system \eqref{eq:syst} with the controller \eqref{eq:usls}, \eqref{eq:dsls} and with $A_t$, $B_t$, $M_t$, $R_t$ being bounded matrix sequences. If for all $t$ holds 
\begin{align}\label{eq:thmasum}
\quad \sum^{T}_{k=1} \left\|\Delta_k\left(A_t,B_t,R_t,M_t\right)\right\|\leq \lambda 
\end{align} and 
\begin{align}\label{eq:adaprule}
    \left\|\sum^{T-1}_{k=1} \left(R_{t-1} -R_{t}\right)(k+1)\hat{\delta}_{t-k}\right\| \leq m_a
\end{align}
 with $\Delta_t$ defined as in \eqref{eq:Ddef}, then the bound $\left\|\hat{\delta}_t\right\| \leq \gamma_t$ is guaranteed, where 
\begin{align}
\label{eq:ineqstab}&\gamma_t = \left(\sqrt[T]{\lambda}\right)^t \left\|x_0\right\| + \frac{1-\lambda^t}{1-\lambda}\left(\hat{\eta}  +m_a\right) &\text{if }\lambda <1 \\
\label{eq:inequnstab}&\gamma_t = \lambda^t \left\|x_0\right\| + \frac{1-\lambda^t}{1-\lambda}\left(\hat{\eta} +m_a\right) &\text{if }\lambda\geq 1
\end{align}
and $x_t$ and $u_t$ are bounded by
\begin{align}
\left\|x_t \right\| &\leq \sum^{T-1}_{k=1} \left\|R_{t}(k)\right\|\max_{\stackrel{t-T \leq i \leq t-1}{i\geq 0}}\gamma_{i} \\
\left\|u_t \right\| & \leq \sum^{T-1}_{k=1} \left\|M_{t}(k)\right\|\max_{\stackrel{t-T \leq i \leq t-1}{i\geq 0}}\gamma_{i}
\end{align}
\end{thm}
\begin{proof}
The proof follows by \lemref{lem:scalar} and the relationship \eqref{eq:usls} and \eqref{eq:xdef}. 
\end{proof}
\noindent We will call $\lambda$ for which the conditions in \thmref{thm:slst} are satisfied, a \textit{robustness margin} of the closed loop. Trivially, if such $\lambda$ is smaller than $1$, \thmref{thm:slst} proves stability of the closed loop:
\begin{coro}
The closed loop is stable, if the conditions in \thmref{thm:slst} are satisfied for a $\lambda <1$. 
\end{coro}
\noindent Furthermore, the next proposition shows we can verify condition \eqref{eq:thmasum} over polytopes $\Pt{}$, by checking \eqref{eq:thmasum} over the extreme points $\mathcal{E}\left(\Pt{}\right)$ of $\Pt{}$.
\begin{prop}\label{prop:polycheck}
Let $\c{P}$ be a convex set and let $\mathcal{E}\left(\Pt{}\right)$ denote the set of extreme points of $\c{P}$. Then, the following equivalence holds:
\begin{align}
&&   \notag \sum^{T}_{k=1} \left\|\Delta_k(A,B,R_t,M_t)\right\| &\leq \lambda,\quad \forall \left[A,B\right] \in \mathcal{E}\left(\c{P}\right)\\
\label{eq:suffcond}\Leftrightarrow&& \sum^{T}_{k=1} \left\|\Delta_k(A,B,R_t,M_t)\right\| &\leq \lambda,\quad \forall \left[A,B\right] \in \c{P}
\end{align}
\end{prop}
\begin{proof}
This follows directly by convexity of $\Pt{}$ and convexity of the function
\begin{align}
  \sum^{T}_{k=1} \left\|\Delta_k(A,B,R,M)\right\|
\end{align}
in $A,B$ for fixed $R$ and $M$. 
\end{proof}

\noindent Combining this proposition with our previous theorem, gives us the following interesting corollary on robustness of time-invariant controller in SLS implementation:
\begin{coro}\label{coro:stabintRM}
If condition \eqref{eq:suffcond} holds for a polytope $\c{P}$ and some fixed $R$ and $M$, then the time-invariant SLS controller 
\begin{align*}
&\hat{\delta}_t = x_t-\sum^{T-1}_{k=1} R(k+1)\hat{\delta}_{t-k} &u_t = \sum^{T-1}_{k=0} M(k+1)\hat{\delta}_{t-k}
\end{align*}
stabilizes \textit{any} time-varying system \eqref{eq:syst} with $\left[A_t,B_t\right] \in \c{P}$.
\end{coro}
\noindent In the next section, \propref{prop:polycheck} will prove useful to compute robust choices of $R_t$ and $M_t$ for polytope uncertainties in the system parameters.

\section{A Scheme for Robust and Adaptive Control with SLS Implementations}\label{sec:RASLS}
\noindent We will combine the findings in \secref{sec:poly} and \secref{sec:slst} to propose a robust adaptive control scheme that can simultaneously learn and control linear systems and will state robustness results for the resulting closed loop. To focus on the main ideas, we will first introduce the proposed control framework for the case of a single linear system. The later section will show an extension of the technique that addresses the general setting of our problem statement \secref{sec:prob} and satisfies the mentioned requirements for scalable implementation.

\subsection{Single System Case}\label{subsec:singlesys}
\noindent For this section we will simplify our problem formulation from \secref{sec:prob} to the case of a single subsystem, i.e. assume
\begin{align*}
&x_{t+1} = Ax_t + Bu_t + w_t &\left\| w_t\right\| \leq \eta
\end{align*}
Combining polytopes of consistent parameters as defined in \defref{def:Ptj} with our previous results from \thmref{thm:slst} and \propref{prop:polycheck} gives rise to the adaptive control algorithm \algoref{algo:RAsls}, which solves the convex optimization problem \eqref{eq:centsls} with the convex cost functions \eqref{eq:flam1} and \eqref{eq:fCD1} at every time-step $t$.
\begin{align}
\label{eq:flam1}f_\lambda(R,M,\lambda) &= \lambda \\
 \label{eq:fCD1}f_{C,D}(R,M,\lambda) &= \sum\limits^{T}_{k=1}\left\|C R(k) + D M(k) \right\|
\end{align}
Moreover, $C$ and $D$ are cost design matrices, $m_a>0$ is a tuning parameter which we will call \textit{adaptation margin} and $\Ex{\Pt{t}}$ denotes the computed set of extreme points of the polytope $\Pt{t}$.\\
\begin{align}
\label{eq:centsls} c_t = \min_{\lambda_t,R_t,M_t}&\quad \quad f(R_t,M_t,\lambda_t)\\
\notag \text{s.t.}&\forall A'\in \c{M}_A(\Ex{\Pt{t}}), B'\in \c{M}_B(\Ex{\Pt{t}}):\\
\label{eq:Deltaconst}& \sum^{T}_{k=1} \left\|\Delta_k(A',B',R_t,M_t)\right\| \leq \lambda_t\\
\label{eq:StableAdaptation}&\left\|\sum^{T-1}_{k=1} \left(R_{t-1} -R_{t}\right)(k+1)\hat{\delta}_{t-k}\right\| \leq m_a \\
 &R_t(1) = I 
\end{align}

\begin{algorithm}
\DontPrintSemicolon
\SetKwInOut{Input}{Input}
\SetAlgoLined
\Input{$\Pt{0}$, $C$, $D$, $m_{a}$, $\lambda^*$, $T$, $\left\{\mathcal{A}_s,\mathcal{B}_s \right\}$}
\BlankLine

	$\lambda_0$, $R_0$, $M_0$ $\leftarrow$ solve \eqref{eq:centsls} with $\Pt{0}$ and $f_\lambda$\; 
	$\hat{\delta}_0 \leftarrow x_0$\;
	apply $u_0 \leftarrow M_0(1)\hat{\delta}_0$\;

\For{ t = 1,2,\dots}{
		compute constraints $\mathcal{C}_t\left(x_t,x_{t-1},u_{t-1}\right)$\;
		update polytope $\Pt{t}\leftarrow\Pt{t-1}\cap\mathcal{C}_t$\;
		$\lambda_t$,$R_t$, $M_t$ $\leftarrow$ solve \eqref{eq:centsls} with $f_\lambda$\;
		\If{$\lambda_t \leq \lambda^*$}{
		$R_t$, $M_t$ $\leftarrow$ solve \eqref{eq:centsls} with $f_{C,D}$ s.t. $\lambda_t \leq \lambda^*$\;
		}
		$\hat{\delta}_t \leftarrow x_t - \left(\sum^{T-1}_{k=1} R_{t}(k+1)\hat{\delta}_{t-k}\right)$\;
		apply $u_t \leftarrow \sum^{T-1}_{k=1} M_{t}(k+1)\hat{\delta}_{t}$\;
}
 \caption{Adaptive Robust Control with SLS}
 \label{algo:RAsls}
\end{algorithm} 

\noindent At the high-level, the adaptive control scheme in \algoref{algo:RAsls} continuously infers new constraints from observations using \eqref{eq:Ctcons} to update the polytope $\Pt{t}$ of consistent parameters (line 5,6) and uses this information in (line 7-10) to find $R_t$ and $M_t$ that satisfy the robustness condition \thmref{thm:slst} at time $t$. Moreover the algorithm finds such robust controllers in two steps. In each iteration, it first searches for $R_t$ and $M_t$ that achieve the smallest robustness margin $\lambda_t$ (line 7) and only if we find feasible controllers that guarantee a minimum desired level of robustness $\lambda^*$, the algorithm re-solves the optimization problem \eqref{eq:centsls} in (line 9) w.r.t. a desired performance objective \eqref{eq:fCD1}. The motivation behind this two-step procedure is clear: Optimizing for a performance objective is only reasonable if robust stability of the closed loop is possible.\\ 

\noindent To verify that algorithm \algoref{algo:RAsls} is well-defined, we need to show that the optimization problem \eqref{eq:centsls} depends only on available information and is feasible for all time. Inspecting the constraints of \eqref{eq:centsls}, we can verify that $R_{t-1}$ and $\hat{\delta}_{t-1}, \dots, \hat{\delta}_{t-T}$ are always available at time $t$ and therefore the optimization problem is well-posed. Furthermore, the problem at time $t=0$ in (Line 1) is always feasible, and by the following proposition \propref{prop:recfeas}, we see that feasibility is guaranteed for all time.
\begin{prop}[Recursive Feasibility]\label{prop:recfeas}
If $R_{t'}, M_{t'}, \lambda_{t'}$ are feasible w.r.t. the constraints of \eqref{eq:centsls} at time $t'$, then $R_{t} = R_{t'}$,$M_{t} = M_{t'}$,$\lambda_{t} = \lambda_{t'}$ are a feasible solution for all time $t \geq t'$. 
\end{prop}
\begin{proof} 
Notice that $\Pt{t} \subset \Pt{t-1}$ and that $R_{t} = R_{t'}$ for $t \geq t'$ eliminates the condition \eqref{eq:StableAdaptation}.
\end{proof}
\begin{prop}[Feasibility implies Robustness]\label{prop:feasrob}
If $R_{t'}, M_{t'}, \lambda_{t'}$ are feasible w.r.t. the constraints of \eqref{eq:centsls} at time $t'$, 
then they satisfy \eqref{eq:Deltaconst} and \eqref{eq:StableAdaptation} for the closed loop at time $t'$.
\end{prop}
\begin{proof}
Condition \eqref{eq:StableAdaptation} and \eqref{eq:adaprule} are the same, so we only have to check \eqref{eq:thmasum}.
 By \propref{prop:polycheck} we now that $R_{t'}, M_{t'}, \lambda_{t'}$ satisfies \eqref{eq:Deltaconst} over the entire polytope $\c{M}_A(\Pt{t'})$ and $\c{M}_B(\Pt{t'})$. Then, by definition of $\Pt{t'}$, the true parameters $A$ and $B$ lie in $\c{M}_A(\Pt{t'})$ and $\c{M}_B(\Pt{t'})$, and therefore $R_{t'}, M_{t'}, \lambda_{t'}$ satisfy \eqref{eq:thmasum}.
\end{proof}
\noindent Moreover, \eqref{eq:StableAdaptation} can be understood as a condition for stable adaptation as we can tune the adaptation margin $m_a$ to control how much $R_t$ can change between time-steps without sacrificing stability.
\propref{prop:recfeas} and \propref{prop:feasrob} give immediate results for robust stability of the closed loop:
\begin{coro}\label{coro:lambdadec}
$\lambda_{t} \leq \max\left\{\lambda_{t-1},\lambda^* \right\} $
\end{coro}
\begin{proof}
The relation $\Pt{t} \subset \Pt{t-1}$ tells us that satisfying condition \eqref{eq:Deltaconst} for all extreme points becomes easier with every iteration step. Then by recursive feasibility \propref{prop:recfeas}, if $\lambda_{t-1} \geq \lambda^*$, then (line 7-10) of \algoref{algo:RAsls} guarantees $\lambda_{t} \leq \max \left\{\lambda_{t-1},\lambda^*\right\}$.
\end{proof}
\begin{coro}\label{coro:robmargin}
 $\lambda_{t'}$ is a robustness margin for the closed loop system for all $t \geq t'$ and the corresponding bounds of \thmref{thm:slst} apply.
\end{coro}
\begin{coro}\label{coro:robstab}
If $\lambda_0 < 1$ in (Line 1) of \algoref{algo:RAsls}, then the closed loop system is stable for all time.
\end{coro}
\noindent  Furthermore, we will call initial uncertainty sets $\Pt{0}$ \textit{strongly stabilizable} if they allow $\lambda_0 <1$ to be solution to the computations of (Line 1).\\
 
\noindent In the next section we will address how the same approach extends to the general problem setting. %

\subsection{Large-Scale System Case}

\noindent Consider now our original problem statement from \secref{sec:prob}, where we try to control the system
\begin{align}\label{eq:sysdist2}
x^{j}_{t+1} & = \sum\limits_{i \in \c{N}(j)}A^{\up{j}{i}} x^i_t + B^j u^j_t + w^j_t.
\end{align}
for potentially large number of subsystems and with the structured uncertainties described by equations  \eqref{eq:AvuBu} and \eqref{eq:previnfo}.\\
As before, the system-wide controller is assumed to have the form \eqref{eq:usls} and \eqref{eq:dsls}. Additionally, by enforcing more sparsity constraints on $R_t$ and $M_t$, we can represent the SLS implementation for subsystem $j$ in the decomposed form:
\begin{align}\label{eq:dlocsls}
u^j_t &= \sum_{i} \sum^{T-1}_{k=0}  \hat{M}^{\up{j}{i}}_{t}(k+1)\hat{\delta}^i_{t-k}\\
\label{eq:deltasls} \hat{\delta}^j_t &= x^j_t+v^j_t-\sum_{i}\sum^{T-1}_{k=1}  \hat{R}^{\up{j}{i}}_{t}(k+1)\hat{\delta}^i_{t-k}
\end{align}
with $\hat{R}^{\up{i}{i}}_t(1) = I$ and $\hat{R}^{\up{j}{i}}_t(1) = 0$ for $i \neq j$. Furthermore, to allow for scalable implementation, we will enforce the following additional design constraints:

\begin{con}[Distributed Computation]\label{con:loccomp}
$\hat{M}^{\up{j}{i}}_{t} := M^{\up{j}{i}}_{t-d_{\up{j}{i}}}$, $\hat{R}^{\up{j}{i}}_{t} := R^{\up{j}{i}}_{t-d_{\up{j}{i}}} $
where $M^{\up{j}{i}}_t$, $R^{\up{j}{i}}_t$ and $\hat{\delta}^i_t$ are computed locally in subsystem $i$ and are broadcasted to the corresponding subsystem $j$ with a delay of $d_{\up{j}{i}}$.
\end{con}
\begin{con}[Localization and Communication Constraints]\label{con:loc}
For every subsystem $i$ define a local region $\c{L}(i) \subset \c{S}(i)$ and enforce the constraints 
\begin{align}\label{eq:locregRM}
\forall j \notin \c{L}(i), \,k:&&M^{\up{j}{i}}_t(k) &= 0 
&R^{\up{j}{i}}_t(k)&= 0,\\
\label{eq:locdelRM}
\forall k < d_{\up{j}{i}}:&& M^{\up{j}{i}}_t(k) &= 0.
&R^{\up{j}{i}}_t(k+1)&= 0
\end{align}
\end{con}
\noindent Under \conref{con:loccomp} and \conref{con:loc}, the implementation \eqref{eq:dlocsls}, \eqref{eq:deltasls} can be verified to satisfy our previously discussed implementation constraints \conref{con:delayintro}, \conref{con:locintro} and \conref{con:loccompintro} for this problem setting. \conref{con:loccomp} and condition \eqref{eq:locdelRM} assure that the communication delays in \conref{con:delayintro} are respected, since the computation of $u^j_t$ and $\hat{\delta}^j_t$ is ensured to depend only on information that is available to subsystem $j$ at time $t$. Moreover, although \conref{con:loccomp} requires every subsystem $i$ to send information to other subsystems, condition \eqref{eq:locregRM} and $\c{L}(i) \subset \c{S}(i)$ from \conref{con:loc} show that the required communication respects \conref{con:locintro}. Furthermore \conref{con:loccompintro} is being addressed, by assuming distributed computation with \conref{con:loccomp} and by condition \eqref{eq:locdelRM} which restricts the number of decision variables for each subsystem $i$ to the size of the region $\c{L}(i)$.

\noindent Following the same approach as in \secref{sec:RASLS}, we can derive the dynamics for the effective disturbance $\hat{\delta}^j$ and obtain the equation \eqref{eq:ddyndist}.
\begin{figure*}[b]
\begin{multline}\label{eq:ddyndist}
\hat{\delta}^j_t= \sum\limits_{i} \left[-\sum^{T}_{k=1} \Delta^{\up{j}{i}}_{k,t-1}\hat{\delta}^{i}_{t-k} -\sum^{T}_{k=1} \left(\Delta^{\up{j}{i}}_{k,t-1-d_{\up{j}{i}}}-\Delta^{\up{j}{i}}_{k,t-1}\right)\hat{\delta}^{i}_{t-k}
+ \sum^{T-1}_{k=1} \left(R^{\up{j}{i}}_{t-1-d_{\up{j}{i}}} -R^{\up{j}{i}}_{t-d_{\up{j}{i}}}\right)(k+1)\hat{\delta}^{i}_{t-k}\right] 
  + \hat{w}^j_t
\end{multline}
\begin{align}\label{eq:im1}
\forall 0 \leq h \leq \bar{d}_i -1: && \sum_{j: d_{\up{j}{i}} \geq h+1} \left\|\sum^{T}_{k = d_{\up{j}{i}}+1} \Delta^{j}_{k}\left(A,B,R^{\up{j}{i}}_t-R^{\up{j}{i}}_{t+h-d_{\up{j}{i}}},M^{\up{j}{i}}_t-M^{\up{j}{i}}_{t+h-d_{\up{j}{i}}} \right)\hat{\delta}^{i}_{t+h+1-k}\right\| &\leq m^{i}_1\\
\label{eq:im2}
\forall 0 \leq h \leq \bar{d}_i -1: && \sum_{j: d_{\up{j}{i}} = h+1} \left\|\sum^{T-1}_{k = h+2} \left(R^{\up{j}{i}}_t-R^{\up{j}{i}}_{t-1}\right)(k+1)\hat{\delta}^{i}_{t+h+1-k}\right\| &\leq m^{i}_2
\end{align}
\end{figure*}
We use the abbreviations $\Delta^{\up{j}{i}}_{k,t}$,$\hat{w}^j_t$ and the function $\Delta^{j}_{k}$ to simplify notation:
\begin{align}
\label{eq:Delij}&\Delta^{\up{j}{i}}_{k,t} :=\Delta^{j}_{k}\left(A,B,R^{\up{j}{i}}_{t},M^{\up{j}{i}}_{t} \right):= \dots \\
 \notag&\dots R^{\up{j}{i}}_t(k+1) -\sum\limits_{n \in \c{N}(j)}A^{\up{j}{n}}R^{\up{n}{i}}_t(k)-B^jM^{\up{j}{i}}_t(k)\\
\label{eq:whatj}
&\hat{w}^j_t := v^j_t -\sum\limits_{i}A^{\up{j}{i}}v^{i}_{t-1} + w^j_{t-1}.
\end{align}
Now, with similar ideas to \secref{sec:RASLS} we can design a distributed controller that guarantees boundedness of the effective disturbances $\hat{\delta}^j$. With the following assumption, we obtain the robustness result \thmref{thm:distsls}:
\begin{asp}\label{asp:propspeed}
The communication speed between subsystems is faster than the propagation speed of disturbances, i.e.:
$d_{\up{j}{i}} < 1+d_{\up{k}{i}},\quad \forall k \in \c{N}(j)$
\end{asp}
\begin{rem}
This is a common assumption in the distributed control community and known to not be very restrictive.
\end{rem}

\begin{thm}\label{thm:distsls}
Given \assref{asp:propspeed}, some $\lambda$, $\rho$ and the local parameters $m^{i}_1$ and $m^{i}_2$, assume that for all $i$, $R^{\up{j}{i}}_t$ and $M^{\up{j}{i}}_t$ satisfy the conditions \eqref{eq:im2}, \eqref{eq:im1} and 
\begin{align}
\label{eq:irob}\left\|\sum\limits_{j \in \c{L}(i)} \Delta^{j}_{k}(A,B,R^{\up{j}{i}}_t,M^{\up{j}{i}}_t)\right\| &\leq c_i \rho^{k-1}\\
\label{eq:ilam}\lambda_i = \sum\limits^{T}_{k=1} c_i \rho^{k-1} &\leq \lambda
\end{align}
for all times $t$, where $\bar{d}_i = \max_{j}{d_{\up{j}{i}}}$. Then the following inequality holds
\begin{align*}
\sum\limits_{j}\left\|\hat{\delta}^j_{t}\right\| \leq \lambda \max\limits_{1\leq k\leq T} \sum\limits_{j}\left\|\hat{\delta}^j_{t-k}\right\| + \sum\limits^{N}_{i=1} \left(m^{i}_1 + \bar{d}_i m^{i}_2 + \hat{w}^i_{t}\right)
\end{align*}
and the bounds of \lemref{lem:scalar} apply for $\sum\limits_{j}\left\|\hat{\delta}^j_{t}\right\|$.
\end{thm}
\begin{proof}
The proof follows the same ideas as in \secref{sec:RASLS}, but for the interest of space we will only provide a sketch. Using the definition \defref{eq:Delij} and \assref{asp:propspeed}, we can verify 
\begin{align}\label{eq:dropdelta}
\Delta^{\up{j}{i}}_{k,t} = 0, \quad k < d_{\up{j}{i}}
\end{align}
which follows by combining \eqref{eq:locdelRM} and \assref{asp:propspeed}. We then proceed to derive an expression for $\sum_{j} \left\|\hat{\delta}^j_{t} \right\|$ by summing up the $N$ equations of the form \eqref{eq:ddyndist} and by use of the triangle inequality we obtain the bound
\begin{align}
\sum\limits_{j}\left\|\hat{\delta}^j_{t+1} \right\| \leq H_{t} + J_{t} + K_{t} + \sum\limits_{j} \left\|\hat{w}^j_{t+1} \right\|
\end{align}
with the three terms $H_{t}$, $J_{t}$ and $K_{t}$:
\begin{align}
H_{t} &=\sum^{T}_{k=1}\sum_{i,j}\left\|\Delta^{\up{j}{i}}_{k,t}\right\| \left\|\hat{\delta}^{i}_{t+1-k} \right\|\\
\label{eq:Jt}J_{t} &= \sum_{j,i}\sum^{T}_{k=1}\left\| \left(\Delta^{\up{j}{i}}_{k,t-d_{\up{j}{i}}}-\Delta^{\up{j}{i}}_{k,t}\right)\hat{\delta}^{i}_{t+1-k}\right\|\\
\label{eq:Kt}K_{t} &= \sum_{j,i} \left\|\sum^{T}_{k=2} \left(R^{\up{j}{i}}_{t-d_{\up{j}{i}}} -R^{\up{j}{i}}_{t+1-d_{\up{j}{i}}}\right)(k)\hat{\delta}^{i}_{t+2-k} \right\|
\end{align}
Now using the condition \eqref{eq:irob} and \eqref{eq:ilam}, we can conclude
\begin{align}
\notag&H_{t}\\
\notag&\leq \sum^{T-1}_{k=0}\sum_{i,j}\left\|\Delta^{\up{j}{i}}_{k+1,t}\right\| \left\|\hat{\delta}^{i}_{t-k} \right\| \leq \sum^{T-1}_{k=0} \sum_{i}c_i\rho^{k} \left\|\hat{\delta}^{i}_{t-k} \right\|\\
& \leq \sum^{T-1}_{k=0}\max\limits_{s}c_s\rho^{k} \sum_{i} \left\|\hat{\delta}^{i}_{t-k} \right\| \leq \lambda \max\limits_{0\leq k\leq T-1} \sum_{i}\left\|\hat{\delta}^{i}_{t-k}\right\|.
\end{align}
After some tedious manipulations of indices and using \assref{asp:propspeed} and \eqref{eq:locdelRM}, it can be established that \eqref{eq:im1} and \eqref{eq:im2} assure the bounds:
\begin{align}\label{eq:recfeas}
&J_{t} \leq \sum\limits^{N}_{i=1} m^{i}_1  &K_{t} \leq \sum\limits^{N}_{i=1} \bar{d}_i m^{i}_2
\end{align}
\end{proof}

\noindent Analogously to the previous simpler case \algoref{algo:RAsls}, we can combine polytopes of consistent parameters as defined in \defref{def:Ptj} with the robustness results \thmref{thm:distsls} and obtain the distributed localized robust adaptive control (DLAR) algorithm \algoref{algo:DLAR}. Similar to \algoref{algo:RAsls}, every subsystem is constructing consistent polytopes $\Pt{t}^j$ for their local parameters (line 8-11) and optimizing for robust $R^j_t$ and $M^j_t$ that can satisfy the conditions of \thmref{thm:distsls}. Depending on the local robustness margin $\lambda^i_t$, the optimization objective is switched between $f_\lambda$ and $f_{C^i,D^i}$ (Line 13-15). After the local controllers have been computed, the control actions and local constraints are broadcasted to the local region of subsystems (Line 18-19).

\begin{align}
\label{eq:distsls} \min_{R^{i}_t,M^{i}_t,c^i_t,\lambda^i_t}&\quad \quad f(R^{i}_t,M^{i}_t,\lambda^i_t) \\ 
\text{s.t.}&\forall A^{\up{q}{p}}\in \c{M}^{\up{q}{p}}_A(\Ex{\Pt{t}^i}), B^{p} \in \c{M}^{p}_B(\Ex{\Pt{t}^i}): \notag\\
& \notag \text{holds }  \eqref{eq:irob}, \eqref{eq:ilam}, \eqref{eq:im1}, \eqref{eq:im2}, \eqref{eq:locdelRM},\eqref{eq:locregRM} \\
& R^{\up{i}{i}}_t(1) = I\text{ and }R^{\up{j}{i}}_t(1) = 0\text{ for }i \neq j
\end{align}



\begin{align}
\label{eq:flam}f_\lambda(R,M,\lambda) &= \lambda \\
\label{eq:fCD}f_{C^i,D^i}(R^i,M^i,\lambda) &= \sum\limits^{T}_{k=1}\left\|C^{i} R^{i}(k) + D^{i}M^{i}(k) \right\|
\end{align}
\begin{algorithm}
\DontPrintSemicolon
\SetKwInOut{Input}{Input}
\SetAlgoLined
\Input{$\Pt{0}^{i}$, $C^i$, $D^{i}$, $m_{1}$, $m_{2}$, $\rho$, $T$, $\mathcal{A}^{\up{j}{i}}_s$, $\mathcal{B}^{i}_s$, $\lambda^*$}
\BlankLine
\For{ subsystem $j = 1:N$}{
	$R^j_0$, $M^j_0$ $\leftarrow$ solve \eqref{eq:distsls} with $\Pt{0}^{j}$\; 
	$\hat{\delta}^j_0 \leftarrow x^j_0$\;
	apply $u^j_0 \leftarrow M^j_0(1)\hat{\delta}^j_0$\;
}
\For{ t = 1,2,\dots}{
	\For{subsystem $i = 1:N$}{		
		$\hat{\mathcal{C}}^{i \leftarrow j}_t \leftarrow \mathcal{C}^{j}_{t-d_{\up{i}{j}}}$ \eqref{eq:Pi_t}\; 
		receive $\hat{R}^{i \leftarrow j}_t$,$\hat{M}^{i \leftarrow j}_t$ $\leftarrow$ $R^{i \leftarrow j}_{t-d_{\up{i}{j}}}$, $M^{i \leftarrow j}_{t-d_{\up{i}{j}}}$ \;
		compute $\mathcal{C}^i_t$ from \eqref{eq:Ctcons}\; 
		update $\Pt{t}^{i}$ $\leftarrow$ $\Pt{t-1}^{i}\cap \bigcap_{j \in \c{R}(i)} \hat{\mathcal{C}}^{i \leftarrow j}_t$ \eqref{eq:Pi_t}\;
		$\lambda^i_t$,$R^i_t$, $M^i_t$ $\leftarrow$ solve \eqref{eq:distsls} with $f_\lambda$ \eqref{eq:flam}  \;
		\If{$\lambda^i_t \leq \lambda^*$}{
		$R^i_t$, $M^i_t$ $\leftarrow$ solve \eqref{eq:distsls} with $f_{C^i,D^i}$ \eqref{eq:fCD} such that $\lambda^i_t \leq \lambda^*$\;
		}
		compute $\hat{\delta}^i_t$,$u^{i}_t$ $\leftarrow$ \eqref{eq:deltasls}, \eqref{eq:dlocsls} \;	
		apply $u^{i}_t$\;
		broadcast $R^{n \leftarrow i}_{t}$, $M^{n \leftarrow i}_{t}$ to all $n \in \mathcal{L}(i)$ \;
		broadcast constraints $\mathcal{C}^i_t$ to all $n \in \mathcal{S}(i)$\;
	}
}
 \caption{A DLAR Control Scheme with SLS}
 \label{algo:DLAR}
\end{algorithm}

\noindent Well-definition, feasibility and robustness of \algoref{algo:DLAR} are derived analogously to our discussion in \secref{subsec:singlesys} and we will cover these aspects only in summary for the interest of space.
It can be verified that the problem \eqref{eq:distsls} in \algoref{algo:DLAR} at time $t$ only depends on information that is available to subsystem $i$ and feasibility for all times follows by \propref{prop:recfeasdist}:
\begin{prop}[Recursive Feasibility]\label{prop:recfeasdist}
If $R^{\up{j}{i}}_{t'}$, $M^{\up{j}{i}}_{t'}$, $c^i_{t'}$ and $\lambda^i_{t'}$ are feasible w.r.t. the constraints of  \eqref{eq:distsls} at time $t'$, then they are guaranteed to satisfy the same constraints for all future times $t\geq t'$.
\end{prop}
\begin{rem}
\assref{asp:propspeed} plays a crucial role for recursive feasibility. In fact, \assref{asp:propspeed} implies \eqref{eq:dropdelta}, which enabled us to derive the recursively feasible conditions \eqref{eq:im1} and \eqref{eq:im2} that guarantee boundedness of the terms \eqref{eq:Jt} and \eqref{eq:Kt}.
\end{rem}
\begin{prop}[Feasibility implies Robustness]\label{prop:feasrobdist}
If $R^{\up{j}{i}}_{t'}$, $M^{\up{j}{i}}_{t'}$, $c^i_{t'}$ and $\lambda^i_{t'}$ are feasible w.r.t. the constraints of  \eqref{eq:distsls} at time $t'$,
then they satisfy condition \eqref{eq:irob}, \eqref{eq:ilam}, \eqref{eq:im1}, \eqref{eq:im2} of \thmref{thm:distsls} for the closed loop at time $t'$.
\end{prop}

\noindent \propref{prop:recfeas} and \propref{prop:feasrob} give immediate results for robust stability of the closed loop:
\begin{coro}\label{coro:lambdadecdist}
$\lambda^{i}_{t} \leq \max\left\{\lambda^{i}_{t-1},\lambda^* \right\} $
\end{coro}
\begin{coro}\label{coro:robmargindist}
 $\max_i \lambda^{i}_{t'}$ is a robustness margin for the closed loop system for all $t \geq t'$ and the corresponding bounds of \thmref{thm:slst} apply.
\end{coro}
\begin{coro}\label{coro:robstabdist}
If $\max_i \lambda^{i}_{0} < 1 $ from (line 2) of \algoref{algo:DLAR}, then the closed loop system is stable for all time.
\end{coro}
 
\noindent In addition, the design constraint \eqref{eq:locdelRM} in \conref{con:loc} constrains the number of decision variables to the size of $\c{L}(i)$ and limits the complexity of the subproblem that every subsystems solves.
\begin{coro}[Local Models]\label{coro:locmod}
The conditions \eqref{eq:im1}, \eqref{eq:im2}, \eqref{eq:irob} and \eqref{eq:ilam} w.r.t. the subsytem $i$ only depend on parameters $A^{\up{v}{u}}$ and $B^{u}$ for which $u\in \c{L}(i)$.
\end{coro}
\begin{proof}
This follows by the definition \eqref{eq:Delij} and using assumption \assref{asp:propspeed}.
\end{proof}
\noindent  \corref{coro:locmod} shows that complexity of the subproblems only grows with the size of the local regions $\c{R}(i)$ and $\c{L}(i)$. This addresses the design constraint \conref{con:loccomp} and shows that \algoref{algo:DLAR} allows for a scalable implementation even for large number $N$ of subsystems.

\section{Simulation}
\noindent The algorithm \algoref{algo:DLAR} is applied to the exemplary control problem with implementation constraints discussed in \secref{sec:introex}. For this simulation we picked the true parameters $\alpha_1 = 0.3$, $\alpha_2 = 0.6$, $\alpha_3 = 0.2$ which produce an unstable open loop system ($\max_i|\lambda_i(A)| = 1.05$). 
\figref{fig:poly} and \figref{fig:truesys} show simulation results of the presented adaptive controller scheme with $\rho = 0.7$, $T = 8$, $x_0 = [0,3,3,3,0]^T$, with respect to two different cases of initial available information. In \figref{fig:poly} the controller has only knowledge of the initial entry-wise bounds on $\alpha$ described in \secref{sec:introex}, while in \figref{fig:truesys} the controller starts off with perfect knowledge of $\alpha$. Furthermore, for the adaptive case, \figref{fig:margindist} summarizes the effective disturbances $\hat{\delta}^j_t$, the environment disturbances $w^j_t$ and individually computed margins  $\lambda^j_t$ for every subsystem. In addition, the quantity $\mu_t$ in \figref{fig:margindist} computes the true robustness margin $\mu_t = \sum_k \left\|\Delta(k)\right\|_1$ of the closed loop adaptive controller w.r.t. to true plant\footnote{Note, that this information is not available to the controllers and is only displayed to show that the controller achieves robust stability.}. Although the initial uncertainty provides a large robustness margin ($\max_i \lambda^i_0 = 4$), the plot of $\mu_t$ in \figref{fig:margindist} shows that the controller learns enough by time-step $t=20$ to render the closed loop stable.

Moreover, even though the controller in \figref{fig:truesys} has perfect knowledge of the parameters, its computed robustness margin is $\sum_k \left\| \Delta (k)\right\|_1 = 0.33$, which tells us that even in presence of full system knowledge, the communication constraints only allow for approximate localization. Putting this in relation to the simulation results in \figref{fig:poly} shows us that the adaptive controller is performing quite well despite large initial uncertainty ($\lambda_0>>1$), communication/localization constraints and decentralized implementation. Although this observation is empirical at this point it shows that \algoref{algo:DLAR} is a promising approach even in the case of large parameter uncertainties.

\begin{figure*}[ht]
\centering
\subcaptionbox{}{\includegraphics{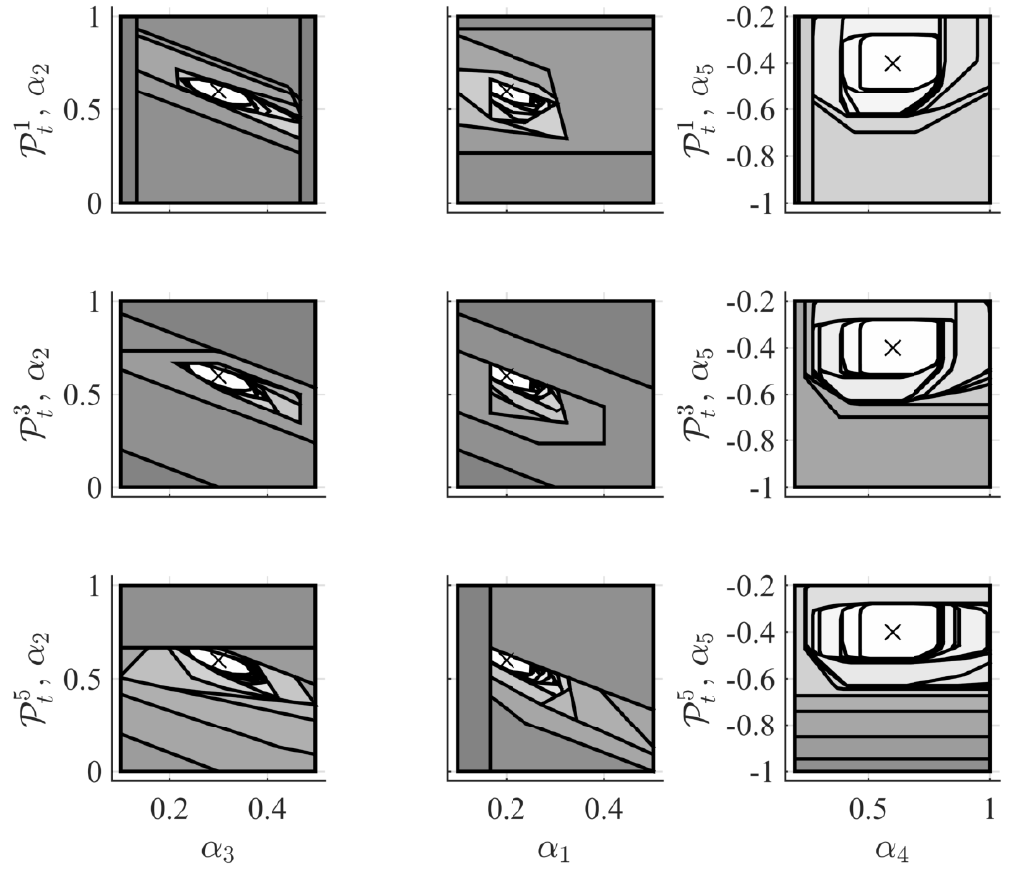}}
\subcaptionbox{}{\includegraphics{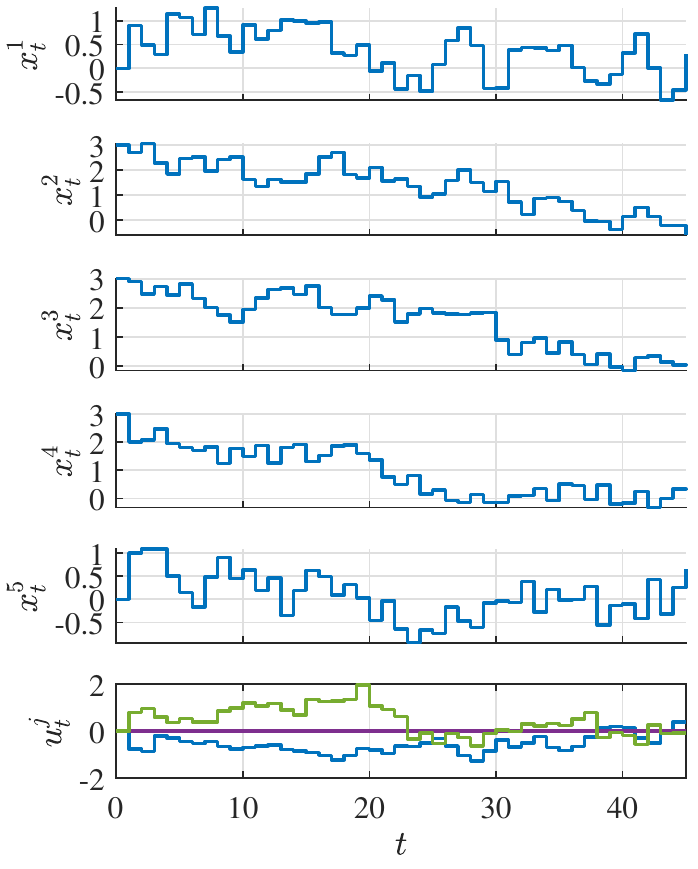}}
\caption{Left: Overlay of projections of $\Pt{t}^j$ onto different coordinates for different time steps. Each row corresponds to a different subsystem $x_1$, $x_3$, $x_5$ (top to bottom). Shading indicates time of computation with shades lightening as simulation time passes. Right: state and input trajectories of closed loop simulation with \algoref{algo:DLAR} and uncertainties on $\alpha$ described in \secref{sec:introex}.}
\label{fig:poly}
\end{figure*}

\begin{figure}[ht]
\centering
\subcaptionbox{}{\includegraphics{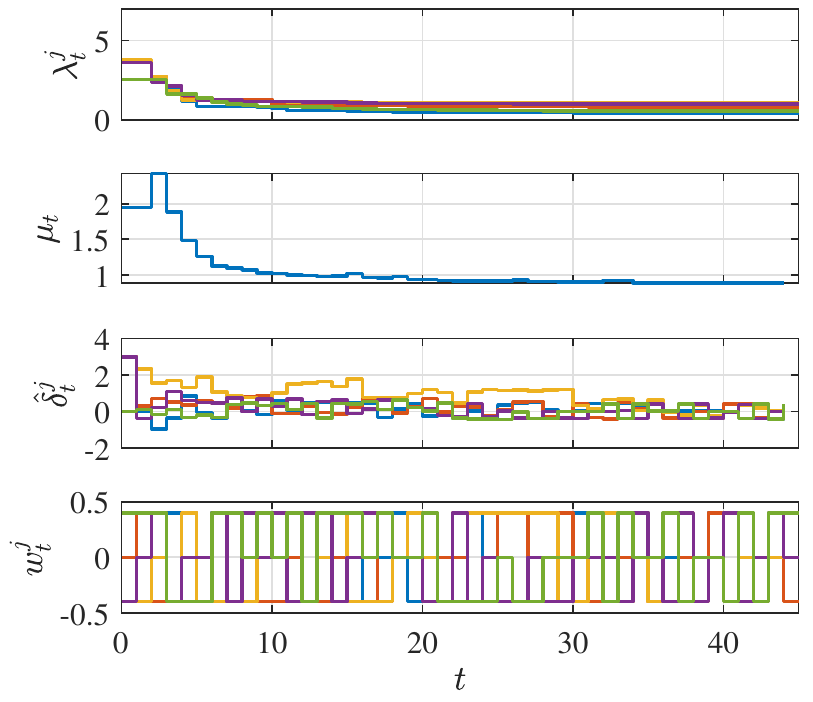}}
\caption{Computed margins $\lambda^j_t$, effective disturbances $\hat{\delta}^j_t$ and disturbance $w^j_t$ for every node $x^j$. $\mu_t = \sum\limits_{k} \left\| \Delta_t(k)\right\|_1$ is computed with the real system and the controller at $t$.} 
\label{fig:margindist}
\end{figure}

\begin{figure}[ht]
\centering
\subcaptionbox{}{\includegraphics{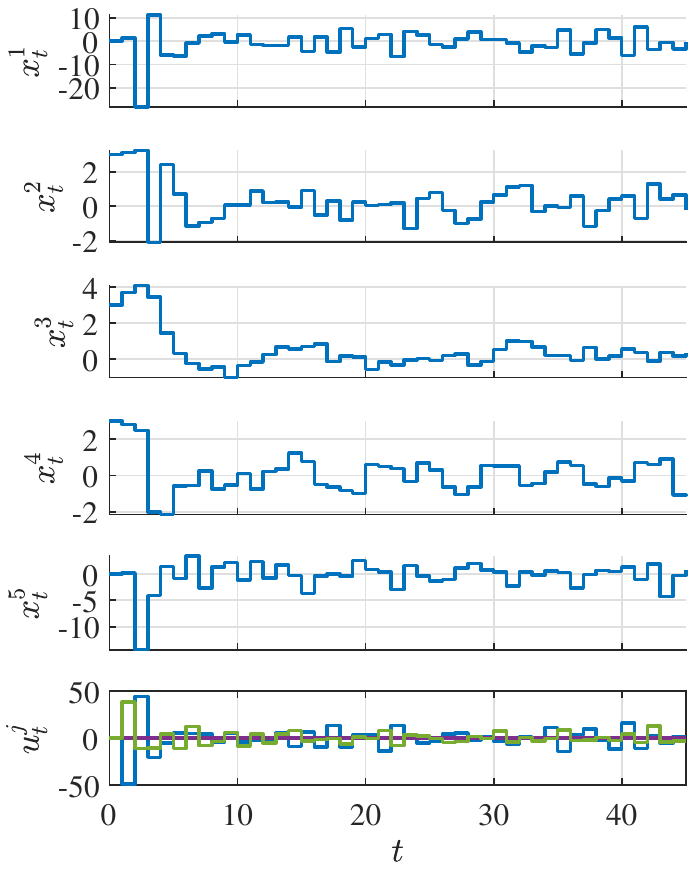}}
\caption{State and input trajectories for closed loop simulation of controller \algoref{algo:DLAR} with perfect parameter knowledge $\alpha$.}
\label{fig:truesys}
\end{figure}

\section{Conclusion and Future Work}
\noindent In this work, we derived a novel framework for adaptive and robust control of linear time-invariant systems. Using the new SLS framework \cite{slsacc} we derived time-varying robustness results which can be used as a new way to design stable adaptations in control systems. With this result we develop a robust adaptive control scheme for linear systems with state feedback under bounded parameter uncertainties, disturbances and noise. The resulting control system continuously infers polytopes of parameters that are consistent with the collected observations and use these sets to compute new robust controllers that improve control performance. In particular, inference of uncertainty sets is done efficiently, since structural properties of the system matrices are directly exploited and subsystems only need to model the dynamics in their local region. Moreover, we present how this approach can incorporate communication constraints and allows for a distributed and scalable control implementation. For the case of small initial uncertainties, a stability proof and worst-case bounds are provided for the closed loop. Finally, simulations with a chain-system empirically show that performance does not degrade too much even if we have large initial uncertainties in the parameters.

Future research will be focused on deriving performance bounds of this technique when dealing with a broader class of uncertainties and reducing the computational cost of the optimization procedures needed in the algorithm.







\bibliographystyle{IEEEtran}
\bibliography{IEEEabrv,refs}

\end{document}